\documentclass[11pt,a4paper]{article}
\usepackage[centertags]{amsmath}
\usepackage{amsfonts}
\usepackage{amssymb}
\usepackage{amsthm}

\usepackage[margin=24mm]{geometry}

\usepackage{newlfont} 
\usepackage{graphicx, amssymb, latexsym, amsfonts, amsmath, lscape, amscd,
	amsthm, color, epsfig} 

\usepackage{amsmath}
\usepackage{amsthm}
\usepackage{amscd}
\usepackage{graphicx}
\usepackage[bottom]{footmisc}
\usepackage{tikz}
\usepackage[T1]{fontenc}

\usetikzlibrary {positioning,chains,fit,shapes,calc}
\definecolor {processblue}{cmyk}{0.96,0,0,0}

\theoremstyle{plain}

\usepackage{microtype}
\usepackage{cite}
\usepackage{framed}
\usepackage[framemethod=tikz]{mdframed}
\usepackage{appendix}
\usepackage{graphicx}
\usepackage{color}
\usepackage{varwidth}
\usepackage{wrapfig}
\usepackage[boxed,noend,linesnumbered]{algorithm2e}
\usepackage[noend]{algpseudocode}
\usepackage[textsize=tiny]{todonotes}

\usepackage[normalem]{ulem}

\usepackage{enumitem}
\setitemize{noitemsep,topsep=3pt,parsep=3pt,partopsep=3pt}

\usepackage{xspace}
\usepackage{xifthen}

\definecolor{darkgreen}{rgb}{0,0.5,0}
\definecolor{darkblue}{rgb}{0,0,0.8}
\usepackage{hyperref}
\hypersetup{
	unicode=false,          
	colorlinks=true,        
	linkcolor=darkblue,          
	citecolor=darkgreen,        
	filecolor=magenta,      
	urlcolor=cyan           
}
\RequirePackage[]{silence}
\newcommand{\calP}{\ensuremath{\mathcal{P}}}

\newcommand{\OPT}{\mathsf{OPT}}

\newcommand{\ignore}[1]{}

\algnewcommand\algorithmicswitch{\textbf{switch}}
\algnewcommand\algorithmiccase{\textbf{case}}

\algdef{SE}[SWITCH]{Switch}{EndSwitch}[1]{\algorithmicswitch\ #1\ \algorithmicdo}{\algorithmicend\ \algorithmicswitch}%
\algdef{SE}[CASE]{Case}{EndCase}[1]{\algorithmiccase\ #1}{\algorithmicend\ \algorithmiccase}%
\algtext*{EndSwitch}%
\algtext*{EndCase}%

\newcommand{\CONGEST}{\ensuremath{\mathsf{CONGEST}}\xspace}
\newcommand{\LOCAL}{\ensuremath{\mathsf{LOCAL}}\xspace}

\newcommand{\eps}{\varepsilon}
\renewcommand{\epsilon}{\varepsilon}
\newcommand{\poly}{\operatorname{\text{{\rm poly}}}}

\newcommand{\set}[1]{\left\{#1\right\}}

\DeclareMathOperator{\polylog}{\poly\log}
\DeclareMathOperator{\dist}{dist}

\newcommand{\hide}[1]{}

\newcommand{\FullOrShort}{short}

\ifthenelse{\equal{\FullOrShort}{full}}{
	
	\newcommand{\fullOnly}[1]{#1}
	\newcommand{\shortOnly}[1]{}

}{
	
	\newcommand{\shortOnly}[1]{#1}
	\newcommand{\fullOnly}[1]{}

}

\renewcommand{\phi}{\varphi}

\newtheorem{theorem}{Theorem}[section]

\newtheorem{lemma}{Lemma}[section]
\theoremstyle{definition}

\theoremstyle{remark}

\numberwithin{equation}{section}

\begin{document}

\author{
	Salwa Faour\\
	University of Freiburg\\
   salwa.faour@cs.uni-freiburg.de
	\and
	Fabian Kuhn\\
	University of Freiburg\\
	kuhn@cs.uni-freiburg.de
}
\title{Usage of the \texttt{\textbackslash author} command}

	\title{  {\bf Approximate Bipartite Vertex Cover\\ in the CONGEST Model}}
        \date{}

	\maketitle
	

\begin{abstract} 
	
	  We give efficient distributed algorithms for the minimum vertex
	cover problem in bipartite graphs in the \CONGEST model. From 
	K\H{o}nig's theorem, it is well known that in bipartite graphs the
	size of a minimum vertex cover is equal to the size of a maximum
	matching. We first show that together with an existing
	$O(n\log n)$-round algorithm for computing a maximum matching, the
	constructive proof of K\H{o}nig's theorem directly leads to a deterministic
	$O(n\log n)$-round \CONGEST algorithm for computing a minimum vertex
	cover. We then show that by adapting the construction, we can also 
	convert an \emph{approximate} maximum matching into an \emph{approximate}
	minimum vertex cover. Given a $(1-\delta)$-approximate
	matching for some $\delta>1$, we show that
	a $(1+O(\delta))$-approximate vertex cover can be computed in time
	$O\big(D+\poly\big(\frac{\log n}{\delta}\big)\big)$, where $D$ is
	the diameter of the graph. When combining with known graph
	clustering techniques, for any $\eps\in(0,1]$, this
	leads to a $\poly\big(\frac{\log n}{\eps}\big)$-time deterministic and
	also to a slightly faster and simpler randomized
	$O\big(\frac{\log n}{\eps^3}\big)$-round \CONGEST algorithm for
	computing a $(1+\eps)$-approximate vertex cover in bipartite
	graphs. For constant $\eps$, the randomized time complexity matches
	the $\Omega(\log n)$ lower bound for computing a
	$(1+\eps)$-approximate vertex cover in bipartite graphs even in the \LOCAL model. Our
	results are also in contrast to the situation in general
	graphs, where it is known that computing an optimal vertex cover
	requires $\tilde{\Omega}(n^2)$ rounds in the \CONGEST model and
	where it is not even known how to compute any 
	$(2-\eps)$-approximation in time $o(n^2)$.
	  
\end{abstract}

\section{Introduction \& Related Work}
\label{sec:intro}

In the minimum vertex cover (MVC) problem, we are given an $n$-node graph $G=(V,E)$ and we are asked to find a vertex cover of smallest possible size, that is, a minimum cardinality subset of $V$ that contains at least one node of every edge in $E$. In the distributed MVC problem, the graph $G$ is the network graph and the nodes of $G$ have to compute a vertex cover by communicating over the edges of $G$. At the end of a distributed vertex cover algorithm, every node $v\in V$ must know if it is contained in the vertex cover or not. Different variants of the MVC problem have been studied extensively in the distributed setting, see e.g., \cite{AstrandFPRSU09,bachrach_podc19,yehuda16,Bar-YehudaCMPP20,disc19_optcovering,censorhillel_disc17,GhaffariJN20,goeoes14_DISTCOMP,GrandoniKP08,GrandoniKPS08,lowerbound,nearsighted}. Classically, when studying the distributed MVC problem and also related distributed optimization problems on graphs, the focus has been on understanding the \emph{locality} of the problem. The focus therefore has mostly been on establishing how many synchronous communication rounds are necessary to solve or approximate the problem in the \LOCAL model, that is, if in each round, each node of $G$ can send an arbitrarily large message to each of its neighbors.

\medskip

\noindent\textbf{MVC in the \boldmath\LOCAL model.} 
The minimum vertex cover problem is closely related to the maximum matching problem, i.e., to the problem of finding a maximum cardinality set of pairwise non-adjacent (i.e., disjoint) edges. Since for every matching $M$, any vertex cover has to contain at least one node from each of the edges $\set{u,v}\in M$, the size of a minimum vertex cover is lower bounded by the size of a maximum matching. We therefore obtain a simple $2$-approximation $S$ for the MVC problem by first computing a maximal matching and by defining the vertex cover $S$ as $S:=\bigcup_{\set{u,v}\in M}\set{u,v}$. It has been known since the 1980s that a maximal matching can be computed in $O(\log n)$ rounds by using a simple randomized algorithm~\cite{alon86,itai86,luby86}. The fastest known randomized distributed algorithm for computing a maximal matching has a round complexity of $O(\log\Delta +\log^3\log n)$, where $\Delta$ is the maximum degree of the graph $G$~\cite{barenboim12,rounding}, and the fastest known deterministic algorithm has a round complexity of $O(\log^2\Delta\cdot \log n)$~\cite{rounding}. A slightly worse approximation ratio of $2+\eps$ can even be achieved in time $O\big(\frac{\log\Delta}{\log\log\Delta}\big)$ for any constant $\eps>0$. This matches the $\Omega\big(\min\set{\frac{\log\Delta}{\log\log\Delta},\sqrt{\frac{\log n}{\log\log n}}}\big)$ lower bound of \cite{lowerbound}, which even holds for any polylogarithmic approximation ratio. In \cite{goeoes14_DISTCOMP}, it was further shown that there exists a constant $\eps>0$ such that computing a $(1+\eps)$-approximate solution for MVC requires $\Omega(\log n)$ rounds even for bipartite graphs of maximum degree $3$. By using known randomized distributed graph clustering techniques~\cite{linial93,MPX13}, this bound can be matched: For any $\eps\in (0,1]$, a $(1+\eps)$-approximate MVC solution can be computed in time $O\big(\frac{\log n}{\eps}\big)$ in the \LOCAL model. It was shown in \cite{ghaffari2017complexity} that in fact all distributed covering and packing problems can be $(1+\eps)$-approximated in time $\poly\big(\frac{\log n}{\eps}\big)$ in the \LOCAL model. By combining with the recent deterministic network decomposition algorithm of \cite{polylogdecomp}, the same result can even be achieved deterministially. We note that all the distributed $(1+\eps)$-approximations for MVC and related problems quite heavily exploit the power of the \LOCAL model. They use very large messages and also the fact that the nodes can do arbitrary (even exponential-time) computations for free.

\medskip 

\noindent\textbf{MVC in the \boldmath\CONGEST model.} As the complexity of the distributed minimum vertex cover and related problems in the \LOCAL model is now understood quite well, there has recently been increased interest in also understanding the complexity of these problems in the more restrictive \CONGEST model, that is, when assuming that in each round, every node can only send an $O(\log n)$-bit message to each of its neighbors. Some of the algorithms that have been developed for the \LOCAL model do not make use of large messages and they therefore directly also work in the \CONGEST model. This is in particular true for all the maximal matching algorithms and also for the $(2+\eps)$-approximate MVC algorithm mentioned above. Also in the \CONGEST model, it is therefore possible to compute a $2$-approximation for MVC in $O(\log\Delta+\log^3\log n)$ rounds and a $(2+\eps)$-approximation in $O\big(\frac{\log \Delta}{\log\log\Delta}\big)$ rounds. However, there is no non-trivial (i.e., $o(n^2)$-round) \CONGEST MVC algorithm known for obtaining an approximation ratio below $2$. For computing an optimal vertex cover on general graphs, it is even known that $\tilde{\Omega}(n^2)$ rounds are necessary  in the \CONGEST model~\cite{censorhillel_disc17}. It is therefore an interesting open question to investigate if it is possible to approximate MVC within a factor smaller than $2$ in the \CONGEST model or to understand for which families of graphs, this is possible. The only result in this direction that we are aware of is a recent paper that gives $(1+\eps)$-approximation for MVC in the square graph $G^2$ in $O(n/\eps)$ \CONGEST rounds on the underlying graph $G$~\cite{Bar-YehudaCMPP20}.

\medskip

\noindent\textbf{MVC in bipartite graphs.} 
In the present paper, we study the distributed complexity of MVC in the \CONGEST model for bipartite graphs. Unlike for general graphs, where MVC is APX-hard (and even hard to approximate within a factor $2-\eps$ when assuming the unique games conjecture~\cite{MVC_UGChard}), for bipartite graphs, MVC can be solved optimally in polynomial time. While in general graphs, we only know that a minimum vertex cover is at least as large as a maximum matching and at most twice as large as a maximum matching, for bipartite graphs, K\H{o}nig's well-known theorem~\cite{koenig_diestel,koenig31} states that in bipartite graphs, the size of a maximum matching is always equal to the size of a minimum vertex cover. In fact, if one is given a maximum matching of a bipartite graph $G=(U\cup V,E)$, a vertex cover of the same size can be computed in the following simple manner. Assume that we are given the bipartition of the nodes of $G$ into sets $U$ and $V$ and assume that we are given a maximum matching $M$ of $G$. Now, let $L_0\subseteq U$ be the set of unmatched nodes in $U$ and let $L\subseteq U\cup V$ be the set of nodes that are reachable from $L_0$ over an alternating path (i.e, over a path that alternates between edges in $E\setminus M$ and edges in $M$). It is not hard to show that the set $S:=(U\setminus L)\cup (V\cap L)$ is a vertex cover that contains exactly one node of every edge in $M$. We note that this construction also directly leads to a distributed algorithm for computing an optimal vertex cover in bipartite graphs $G$. The bipartition of $G$ can clearly be computed in time $O(D)$, where $D$ is the diameter of $G$ and given a maximum matching $M$, the set $L$ can then be computed in $O(n)$ rounds by doing a parallel BFS exploration on alternating paths starting at all nodes in $L_0$. Together with the $O(n\log n)$-round \CONGEST algorithm of \cite{ahmadi18} for computing a maximum matching, this directly leads to a deterministic $O(n\log n)$-round \CONGEST algorithm for computing an optimal vertex cover in bipartite graphs. As our main contribution, we show that it is not only possible to efficiently convert an optimal matching into an optimal vertex cover, but we can also efficiently turn an \emph{approximate} solution of the maximum matching problem in a bipartite graph into an \emph{approximate} solution of the MVC problem on the same graph. Unlike for MVC, where no arbitrarily good approximation algorithms are known for the \CONGEST model, such algorithms are known for the maximum matching problem~\cite{ahmadi18,yehuda17,lotker15}. We use this to develop polylogarithmic-time approximation schemes for the bipartite MVC problem in the \CONGEST model. We next discuss our main contributions in more detail.

\subsection{Contributions} 
\label{sec:contributions}

Our first contribution is a simple linear-time algorithm to solve the exact minimum vertex cover problem.

\begin{theorem} \label{thm:exact} 
	There is a deterministic \CONGEST algorithm to (exactly) solve the minimum vertex cover problem in bipartite 
	graphs in time $O(\OPT\cdot \log \OPT)$, where $\OPT$ is the size of a minimum vertex cover.  
\end{theorem}

\begin{proof}
	As mentioned, the algorithm is a straightforward \CONGEST implementation of K\H{o}nig's constructive proof.  Given a bipartite graph $G=(U\cup V,E)$, one first computes a maximum matching $M$ of $G$ in time $O(\OPT\cdot\log\OPT)$ by using the \CONGEST algorithm of \cite{ahmadi18}. One elects a leader node $\ell$ and computes a BFS tree of $G$ rooted at $\ell$ in time $O(D)$, where $D$ is the diameter of $G$. Let $U$ be the set of nodes at even distance from $\ell$ and let $V$ be the set of nodes at odd distance from $\ell$. Let $L_0$ be the set of nodes in $U$ that are not contained in any edge of $M$. Starting at $L_0$, we do a parallel BFS traversal on alternating paths. Let $L$ be the set of nodes that are reached in this way. The set $L$ can clearly be computed in time $O(|M|)=O(\OPT)$. As shown in the constructive proof of K\H{o}nig's theorem~\cite{koenig_diestel,koenig31}, the minimum vertex cover $S$ is now defined as $S:=(U\setminus L) \cup (V\cap L)$.
\end{proof}

Our main results are two distributed algorithms to efficiently compute $(1+\eps)$-approximate solutions to the minimum vertex cover problem. We first give a slightly more efficient (and also somewhat simpler) randomized algorithm.

\begin{theorem}\label{thm:randomized}
For every $\eps\in (0,1]$, there is a randomized \CONGEST algorithm that for any bipartite $n$-node graph $G$ computes a vertex cover of expected size at most $(1+\eps)\cdot\OPT$ in time $O\big(\frac{\log n}{\eps^3}\big)$, w.h.p., where $\OPT$ is the size of a minimum vertex cover of $G$.
\end{theorem}

We remark that for constant $\eps$, the above result matches the lower bound of \cite{goeoes14_DISTCOMP} for the \LOCAL model. More precisely, in \cite{goeoes14_DISTCOMP}, it is shown that there exists a constant $\eps>0$ for which computing a $(1+\eps)$-approximation of minimum vertex cover requires $\Omega(\log n)$ rounds even on bounded-degree bipartite graphs. The second main result shows that similar bounds can also be achieved deterministically.

\begin{theorem}\label{thm:deterministic}
For every $\eps\in (0,1]$, there is a deterministic \CONGEST algorithm that for any bipartite $n$-node graph $G$ computes a vertex cover of size at most $(1+\eps)\cdot\OPT$ in time $\poly \big(\frac{\log n}{\epsilon}\big)$, where $\OPT$ is the size of a minimum vertex cover of $G$.
\end{theorem}

\subsection{Our Techniques in a Nutshell}
\label{sec:nutshell}

We next describe the key ideas that leads to the results in Theorems \ref{thm:randomized} and \ref{thm:deterministic}. The core of our algorithms is a method to efficiently transform an approximate solution $M$ for the maximum matching problem into an approximate solution of MVC. More concretely, assume that we are given a matching $M\subseteq E$ of a bipartite graph $G=(U\cup V,E)$ such that $M$ is a $(1-\eps)$-approximate maximum matching of $G$ (for a sufficiently small $\eps>0$). In Section \ref{sec:diameter}, we then first show that we can compute a vertex cover $S\subseteq U\cup V$ of size $(1+O(\eps\polylog n))\cdot|M|$ (and therefore a $(1+O(\eps\polylog n))$-approximation for MVC) in time $O\big(D+\poly\big(\frac{\log n}{\eps}\big)\big)$, where $D$ is the diameter of $G$. If the matching $M$ has the additional property that there are no augmenting paths of length at most $2k-1$ for some $k=O(1/\eps)$, we show that such a vertex cover $S$ can be obtained by adapting the constructive proof of K\H{o}nig's theorem. Clearly, the bipartition of the nodes of $G$ into sets $U$ and $V$ can be computed in time $O(D)$. Now, we again define $L_0$ as the set of unmatched nodes in $U$ and more generally for any integer $i\in \set{1,\dots,2k}$, we define $L_i$ to be the set of nodes in $U\cup V$ that can be reached over an alternating path of length $i$ from $L_0$ and for which no shorter such alternating path exists. Note that all nodes in set $L_{2j-1}$ for $j\in\set{1,\dots,k}$ are matched nodes as otherwise, we would have an augmenting path of length at most $2k-1$. 
Note that any alternating path starting at $L_0$ starts with a non-matching edge from $U$ to $V$ and it alternates between non-matching edges from $U$ to $V$ and matching edges from $V$ to $U$. For every $j\geq 1$, the set $L_{2j}$ therefore exactly contains the matching neighbors of the nodes in $L_{2j-1}$ and we therefore have $|L_{2j}|=|L_{2j-1}|$. We will show that for every $j\in \set{1,\dots,k}$ the set
\[
S_j := \bigcup_{j'\in\set{1,\dots,j}} L_{2j'-1} \cup \left(U\setminus \bigcup_{j'\in\set{0,\dots,j-1}}L_{2j'}\right)
\]
is a vertex cover of size $|M| + |L_{2j}|=|M|+|L_{2j-1}|$. Because the sets $L_i$ are disjoint, clearly one of these vertex covers must have size at most $\big(1+\frac{1}{k}\big)\cdot|M|=(1+O(\eps))\cdot|M|$.

If we do not have the guarantee that $M$ does not have short augmenting paths, we show that one can first delete $O(\eps\cdot |M|\cdot \polylog n)$ nodes from $U\cup V$ such that in the induced subgraph of the remaining nodes, there are no short augmenting paths w.r.t.\ $M$. We also show that we can find such a set of nodes to delete in time $\poly\big(\frac{\log n}{\eps}\big)$. We can therefore then first compute a good vertex cover approximation for the remaining graph and we then obtain a vertex cover of $G$ by also adding all the removed nodes to the vertex cover. 

Given our algorithm to compute a good MVC approximation in time $O(D+\poly\log n)$ in Section \ref{sec:polylog}, we show how that in combination with known graph clustering techniques, we can obtain MVC approximation algorithms with polylogarithmic time complexities and thus prove Theorems \ref{thm:randomized} and \ref{thm:deterministic}. Given a maximal matching $M$, we show that we can compute disjoint low-diameter clusters such that all the edges between clusters can be covered by $O(\eps\cdot |M|)$ nodes. With randomization, such a clustering can be computed by using the random shifts approach of \cite{blelloch14,MPX13} and deterministically such a clustering can be computed by a simple adaptation of the recent network decomposition algorithm of \cite{polylogdecomp}. Since the clusters have a small diameter, we can then use the algorithm of Section \ref{sec:diameter} described above inside the clusters to efficiently compute a good MVC approximation.

\section{Model and Definitions}
\label{sec:model}

\noindent\textbf{Communication Model:}  We work with
the standard \CONGEST model \cite{peleg00}. The network is modelled as
an $n$-node undirected graph $G =
(V,E)$ with maximum degree at most $\Delta$ and each node has a unique
$O(\log n)$-bit identifier. The computation proceeds in synchronous
communication rounds. Per round, each node can perform some local
computations and send one $O(\log n)$-bit message to each of its
neighbors. At the end, each node should know its own part of the
output, e.g., whether it belongs to a vertex cover or not.

\smallskip

\noindent\textbf{Low-Diameter Clustering:} In order to reduce the
problem of approximating MVC on general (bipartite) graphs to
approximating MVC on low-diameter (bipartite) graphs, we need a
slightly generalized form of a standard type of graph clustering. Let
$G=(V,E,w)$ be a weighted graph with non-negative edge weights $w(e)$
and assume that $W:=\sum_{w\in E} w(e)$ is the total weight of all
edges in $G$. A subset $S\subseteq V$ of the nodes of $G$ is called
\emph{$\lambda$-dense} for $\lambda\in[0,1]$ if the total weight of
the edges of the induced subgraph $G[S]$ is at least $\lambda\cdot W$.
A \emph{clustering} of $G$ is a collection $\set{S_1,\dots,S_k}$ of
disjoint subsets $S_i\subseteq V$ of the nodes. A clustering
$\set{S_1,\dots,S_k}$ is called $\lambda$-dense if the set
$S:=S_1\cup\dots\cup S_k$ is $\lambda$-dense. The \emph{strong
	diameter} of a cluster $S_i\subseteq V$ is the (unweighted) diameter
of the induced subgraph $G[S_i]$ and the \emph{weak diameter} of a
cluster $S_i\subseteq V$ is the maximum (unweighted) distance in $G$
between any two nodes in $S_i$. The strong/weak diameter of a
clustering $\set{S_1,\dots,S_k}$ is the maximum strong/weak diameter
of any cluster $S_i$. A clustering $\set{S_1,\dots,S_k}$ is called
\emph{$h$-hop separated} for some integer $h\geq 1$ if for any two
clusters $S_i$ and $S_j$ ($i\neq j$), we have
$\min_{(u,v)\in S_i\times S_j} d_G(u,v)\geq h$, where $d_G(u,v)$
denotes the hop-distance between $u$ and $v$ in $G$. A clustering
$\set{S_1,\dots,S_k}$ is called \emph{$(c,d)$-routable} if we are in
addition given 
a collection of trees $T_1,\dots,T_k$ in $G$ such that for
every $i\in\set{1,\dots,k}$, the node set of $T_i$ contains the nodes
in $S_i$, the height of $T_i$ is at most $d$ and every edge $e\in E$
of $G$ is contained in at most $c$ trees $T_1,\dots,T_k$. Note that a
$(c,d)$-routable clustering clearly has weak diameter at most
$2d$. Note also that any clustering with strong diameter $d$ can
easily be extended to a $(1,d)$-routable clustering by computing a BFS
tree $T_i$ for the induced subgraph $G[S_i]$ of each cluster $S_i$. 

\section{Approximating MVC in Time Linear in the Diameter}
\label{sec:diameter}
In this section, we show how to compute a minimum vertex cover
approximation in time $O(D+\poly\log n)$, where $D$ is the diameter of
the graph. A bit more precisely, we will show the following. Let
$G=(V,E)$ be a bipartite graph with diameter $D$ and let $G'=(V',E')$ be a
subgraph of $G$. Assume that each node of $G$ knows if it is contained
in the set $V'$ and which of its edges are contained in the set
$E'$. We then show that for any $\eps\in(0,1]$, we can compute a
$(1+\eps)$-approximate vertex cover of $G'$ in time
$O\big(D+\poly\big(\frac{\log n}{\eps}\big)\big)$ in the \CONGEST model on graph $G$.

Given a matching $M$ of any graph $G$, a path is said to be augmenting
w.r.t $M$ in $G$ if it is a path that starts and ends with unmatched
vertices and alternates between matched and unmatched edges w.r.t.\
$M$ in $G$. Inspired by K\H{o}nig's constructive proof, we first
describe for an integer $k \geq 1 $, a deterministic algorithm that
gives an approximate minimum vertex cover in bipartite graphs from an
approximate maximum matching with the guarantee that no augmenting
paths of length at most $2k-1$ exist in the graph. We will then use
this algorithm as a subroutine in all our subsequent algorithms.  We
remark that a similar but randomized construction has also been used
by Feige, Mansour, and Schapire for the local computation algorithms
model~\cite{FMS15}.

\begin{lemma}\label{lemma:diameter1}
	Let $G=(V,E)$ be a bipartite graph of diameter $D$, let $G'=(V',E')$ be a subgraph of $G$ (i.e., $V'\subseteq V$ and $E'\subseteq E$), and let $k\geq 1$ be an integer parameter. Assume that $M$ is a matching of $G'$ s.t.\ there exists no augmenting path of length at most $2k-1$ w.r.t.\ $M$ in $G'$. Then, there exists a deterministic \CONGEST model algorithm to compute a $(1+1/k)$-approximate minimum vertex cover of $G'$ in $O(D+k)$ rounds on graph $G$.  
\end{lemma}

\begin{proof}
	Let $G=(V,E)$ be a bipartite graph. Let $k\geq 1$ be an integer
	parameter. As a first step, in $O(D)$ rounds, we elect a leader node
	and compute a BFS tree of $G$ rooted at the leader node. By using
	the BFS tree, we also compute the bipartition of $V$ into two
	independent sets in time $O(D)$. Let $A$ and $B$ be the parts of this
	bipartition. Let $M$ be a matching in $G$ such that there exists no
	augmenting paths w.r.t $M$ of length at most
	$2k-1$.
	
	In the following, we use $A'\subseteq A$ and $B'\subseteq B$ to
	denote the subsets of nodes in $A$ and $B$ that are in subgraph
	$G'$. We are now going to partition the sets $A'$ and $B'$. Initially, each unmatched node in $A'$ considers itself in set
	$A'_{0}$, while the remaining nodes $A' \setminus A'_0$ and $B'$ are considered in
	set $A'_{\infty}$ and $B'_{\infty}$, respectively. In the following, some nodes of
	$A'_{\infty}$ and $B'_{\infty}$ will be moved from $A'_{\infty}$ and $B'_{\infty}$ to other sets.
	To compute the partition, we define a directed graph $D(G')$ whose underlying graph is
	$G'$ as follows. In $D(G')$, every matched edge in $G'$ is represented by an arc
	from $B'$ to $A'$ and each unmatched edge in $G'$ is represented by
	an arc from $A'$ to $B'$. We then partition the sets $A'$ and $B'$
	as follows. We build a directed forest of depth $2k$ where all the
	nodes $A_0'$ (i.e., the unmatched $A'$-nodes) are
	the roots by running the first $2k$ iterations of parallel breadth
	first search on $D(G')$ starting from each node in $A'_0$. Then,
	each node in the odd level iteration $j \in \{1,3,...,2k-1\}$ of the
	BFS directed forest switches from $B'_{\infty}$ to
	$B'_{\frac{j+1}{2}}$. Similarly, nodes in the even level
	iteration $h \in \{2,4,...,2k\}$ switch from $A'_{\infty}$ to
	$A'_{\frac{h}{2}}$. Notice that nodes that have not been hit by this BFS
	are still in $A'_{\infty}$ or $B'_{\infty}$. The time required to do
	this parallel BFS and thus obtain the partition of $A'$ and $B'$ is
	$O(k)$ rounds in the \CONGEST model.
	
	By using the partition of $A'$ and $B'$, we can define $k$ different
	vertex covers $C_1',C_2',\dots,C_k'$ of $G'$ as follows. Vertex cover $C_s'$
	is defined as
	\[
	C'_s := A_{\infty}'\cup \bigcup_{i=s}^k A_i' \cup
	\bigcup_{i=1}^{s} B_i'.
	\]
	We first show that indeed each set $C_s'$ is a vertex cover of $G'$. Suppose $C_s'$ is not a
	vertex cover, i.e., there exists an edge $e=\{u,v\}$ such that
	$u \in \bigcup_{i=0}^{s-1} A_{i}'$ and
	$ v \in (\bigcup_{i=s+1}^{k} B_{i}' ) \cup B'_{\infty}$.  W.l.o.g.,
	assume that it is an edge that connects a node in $u\in A_x'$ for $j<s$
	to some node in $v\in B_y'$ for $y>s$. By
	the partitioning scheme of $A'$ and $B'$ such an edge cannot exist. 
	If the edge is unmatched, it is a directed edge from $A'$ to $B'$ in
	$D(G')$ and therefore in the parallel BFS on $D(G')$ $v$ would be
	reachable from $u$ and therefore $v$ would be in $B_{x+1}'$. The
	edge also cannot be a matching edge because in the parallel BFS,
	each node in $A_j'$ has been reached over its matching edge from a
	node in $B_j'$. Hence, $C_{s}'$ is a vertex cover for every $s\in \set{1,\dots,k}$.
	
	To finish the proof, we will show that there exists an
	$i^*\in\set{1,\dots,k}$ such that $|C_{i^*}|\leq (1+1/k)\cdot\OPT$,  where $\OPT$ is the size of the
	minimum vertex cover of $G'$. To prove this, we first observe that all
	the nodes in $B_1',\dots,B_k'$ are matched nodes w.r.t.\ the given
	matching $M$. This follows from
	the fact that the matching $M$ has no augmenting paths of length at
	most $2k-1$. If one of the nodes in $B_1',\dots,B_k'$ is unmatched,
	there is a directed path of length at most $2k-1$ in $D(G')$ from a
	node in $A_0'$ to an unmatched node in $B'$. Such a path corresponds
	to an augmenting path of the same length in $G'$ and therefore
	cannot exist. Because all node in $B_1',\dots,B_k'$ are matched, we
	can further conclude that for every $i\in \set{1,\dots,k}$, we have
	$|A_i'|=|B_i'|$ because the set $A_i'$ is exactly the set of nodes
	that can be reached over the matching edges from the nodes in $B_i'$
	in $G'$. Note also that the matching edge of a node in $B_i'$ cannot go
	to a node in $A_j'$ for $j<i$ because these nodes have their
	matching partners in set $B_j'$. 
	
	We now define $i^*$ as
	$i^* := \arg\min_{i \in \set{1,\dots,k}} B_i'$. The size of the
	vertex cover $C_{i^*}'$ can be bounded as follows:
	\begin{eqnarray*}
		|C_{i^*}'| 
		& = & |A_\infty'| + \sum_{i=i^*}^k |A_i'| + \sum_{i=1}^{i^*}|B_i'|\\
		& = & |M| - \sum_{i=1}^{i^*-1}|A_i'| + \sum_{i=1}^{i^*}|B_i'|\
		=\ |M| + |B'_{i^*}|.
	\end{eqnarray*}
	The second equation follows because all the nodes in $A'\setminus
	A_0'$ are matched nodes and the last equation follows because for
	all $i\in \set{1,\dots,k}$, we have $|A_i'|=|B_i'|$. From the fact
	that the sets $B_i'$ are disjoint and the definition of $i^*$, it
	now directly follows that $|B_{i^*}'|\leq
	\frac{1}{k}\cdot\sum_{i=1}^k|B_i'|\leq \frac{1}{k}\cdot|M|$. We
	therefore have $|C_{i^*}'|\leq (1+1/k)\cdot|M|$.
	
	It remains to show that the time complexity of the algorithm is
	$O(D+k)$ as claimed. We already saw that the partition of the nodes
	of $G'$ into $A_0',\dots,A_k'$, $A_{\infty}'$, $B_1',\dots,B_k'$,
	and $B_{\infty}'$ can be computed in time $O(k)$. To compute the
	vertex cover $C_{i^*}'$, it therefore remains to compute $i^*$. We
	do this, by first computing the sizes of all sets $B_i'$ by aggregating  the
	sums over the already computed BFS spanning tree of $G$. By using a
	standard pipelining argument, the $k$ sums can be computed in time
	$O(D+k)$. The root of the tree can then compute $i^*$ and broadcast
	it along the BFS tree to all nodes in $G'$ in time $O(D)$.
\end{proof} 

In combination with a distributed approximate maximum matching algorithm of Lotker, Patt-Shamir, and Pettie~\cite{lotker15}, Lemma \ref{lemma:diameter1} directly leads to a randomized $O(D+\poly\log n)$-round distributed approximation scheme for the MVC problem.
\begin{theorem}\label{approxMVC:rand}
	Let $G=(V,E)$ be a bipartite graph of diameter $D$ %
	and $G'=(V',E')$ be a subgraph of $G$ (i.e., $V'\subseteq V$ and $E'\subseteq E$). For $\epsilon\in(0,1]$, there is a randomized algorithm that gives a  $(1+\epsilon)$-approximate minimum vertex cover of $G'$ w.h.p.\ in $O(D+ \frac{\log n}{\epsilon^3} )$ rounds in the \CONGEST model on $G$.
\end{theorem}
\begin{proof}
	The approximate maximum matching algorithm of \cite{lotker15} is based on the classic approach of Hopcroft and Karp~\cite{hopkarp}. For a given graph and positive integer parameter $k$, the algorithm computes a matching $M$ of the graph such that there is no augmenting path of length at most $2k-1$ w.r.t.\ $M$. When run on an $n$-node graph, the algorithm w.h.p.\ has a time complexity of $O(k^3\cdot \log n)$ in the \CONGEST model. The theorem therefore directly follows by applying the algorithm of \cite{lotker15} on $G'$ with $k=\lceil 1/\eps\rceil$ and by Lemma \ref{lemma:diameter1}.
\end{proof}

\subsection{Deterministic MVC Approximation}
\label{sec:det_diameter}

The only part in the algorithm underlying Theorem \ref{approxMVC:rand} that is randomized is the approximate maximum matching algorithm of \cite{lotker15}. In order to also obtain a deterministic distributed MVC algorithm, we therefore have to replace the randomized distributed matching algorithm by a deterministic distributed matching algorithm. The algorithm of \cite{lotker15} is based on the framework of \cite{hopkarp} and it therefore guarantees that the resulting matching has no short augmenting paths. While the size of such a matching is guaranteed to be close to the size of a maximum matching, the converse is not necessarily true.\footnote{One can for example obtain an almost-maximum matching $M$ for some graph $G$ by taking a maximum matching of $G$ and flipping an arbitrary matched edge to unmatched. While the matching $M$ is obviously  a very good approximate matching, it has a short augmenting path of length $1$.} Unfortunately, we are not aware of an efficient deterministic \CONGEST model algorithm to compute a matching $M$ with no short augmenting paths. To resolve this issue, we therefore have to do some additional work.

For $\eps >0$, we define an augmenting path w.r.t.\ a matching in $G'$
to be short if it is of length at most $\ell= 2k'-1$, where
$k'= \lceil 2/\eps\rceil$. We define $\delta \leq \epsilon/(2\alpha)$
where $\alpha= O\big(\frac{\log\Delta}{\epsilon^3}\big)$. We first run
a polylogarithmic-time deterministic \CONGEST algorithm by Ahmadi et
al.~\cite{ahmadi18} to obtain a $(1-\delta)$-approximate maximum
matching $M$ in $G'$. This matching $M$ can potentially have short
augmenting paths. In order to get rid of short augmenting paths, we
then find a subset of nodes $S_1$ such that after deleting the nodes
in $S_1$, $M$ is a matching with no short augmenting paths in the
remaining subgraph $G''$ of $G'$. We show that we can select $S_1$
such that $|S_{1}| \leq \alpha \delta \OPT$, where $\OPT$ is the size
of a minimum vertex cover in $G'$. Now that we end up with a matching
in $G''$ with no short augmenting paths, we can directly apply our
subroutine from above on $G''$ and obtain a set $S_2$ which is a
$(1+\frac{\epsilon}{2})$-approximate vertex cover of $G''$. Finally,
we deduce that $C= S_1 \cup S_2$ is a vertex cover of $G'$. Moreover,
since the size of the minimum vertex cover of $G''$ is at most $\OPT$,
we get
$|C|= |S_{1}| + |S_2| \leq \alpha \delta \OPT+
(1+\frac{\epsilon}{2})\OPT= (1+\epsilon)\OPT$.

\medskip

\noindent\textbf{Finding \boldmath$S_1$.}
We next describe an algorithm to compute the set $S_1$. We assume that we are given an arbitrary $(1-\delta)$-approximate matching $M$ of $G'=(U'\cup V', E')$. As discussed above, we need to find a node set $S_1\subseteq U'\cup V'$ that allows to get rid of augmenting paths of length at most $\ell=2k'-1$.  This will be done in $(\ell+1)/2$ stages $d=1,3,\dots,\ell$. The objective of stage $d$ is to get rid of augmenting paths of length exactly $d$. Note that this guarantees that when starting stage $d$, there are no augmenting paths of length less than $d$ and thus in stage $d$, all augmenting paths of length $d$ are also shortest augmenting paths.
In the following, we focus on a single stage $d$. Formally, the subproblem that we need to solve in stage $d$ is the following.

We are given a bipartite graph $H=(U_H\cup V_H,E_H)$ with at most $n$ nodes and we are given a matching $M_H$ of $H$. We assume that the bipartition of the graph into $U_H$ and $V_H$ is given. Let $d$ be a positive odd integer and assume that $H$ has no augmenting paths of length shorter than $d$ w.r.t.\ $M_H$. The goal is to find a set $S_H\subseteq U_H\cup V_H$ that is as small as possible such that when removing the set $S_H$ from the nodes of $H$ and the resulting induced subgraph $H':=H[U_H\cup V_H\setminus S_H]$ has no augmenting paths of length at most $d$ w.r.t.\ the matching $M_H':=M_H\cap E(H')$, i.e., w.r.t.\ to the matching induced by $M_H$ in the induced subgraph $H'$ of the remaining nodes.

We therefore need to find a set $S_H$ of nodes of $H$ such that $S_H$ contains at least one node of every augmenting path of length $d$ w.r.t.\ $M_H$ in graph $H$. Further, we want to make sure that after removing $S_H$, in the remaining induced subgraph $H'$ w.r.t.\ the remaining matching $M_H'$, there are no augmenting paths that were not present in graph $H$ w.r.t.\ matching $M_H$. To guarantee this, we make sure that whenever we add a matched node in $U_H\cup V_H$ to $S_H$, we also add its matched neighbor to $S_H$. In this way, every node that is unmatched in $H'$ was also unmatched in $H$ and therefore any augmenting path in $H'$ is also an augmenting path in $H$.

\medskip

\noindent\textbf{Getting Rid of Short Augmenting Paths by Solving Set Cover.} The problem of finding a minimal such collection of matching edges and unmatched nodes can be phrased as a minimum set cover problem. The ground set $\calP$ is the set of all augmenting paths of length $d$ w.r.t.\ $M_H$ in $H$. For each unmatched node $v\in U_H\cup V_H$, we define $P_v$ as the set of augmenting paths of length $d$ that contain $v$. Similarly, for each matching edge $e\in M_H$, we define $P_e$ as the set of augmenting paths of length $d$ that contain $e$. 
The goal is to find a smallest set $C$ consisting of unmatched nodes $v$ in $U_H\cup V_H$ and matching edges $e\in M_H$ such that the union of the corresponding sets $P_v$ and $P_e$ of paths covers all paths in $\calP$. The set $S_H$ then consists of all nodes in $C$ and both nodes of each edge in $C$. Let us first have a look at the structure of augmenting paths of length $d$ in $H$. Let $L_0$ be the set of unmatched nodes in $U_H$ and more generally let $L_i\subseteq U_H\cup V_H$ for $i\in\set{0,\dots,d}$ be the set of nodes of $H$ that can be reached over a shortest alternating path of length $i$ from a node in $L_0$. Since the bipartition into $U_H$ and $V_H$ is given, the sets $L_0,\dots,L_d$ can be computed in $d$ \CONGEST rounds by a simple parallel BFS exploration. Since we assume that $H$ has no augmenting paths of length shorter than $d$, every augmenting path of length $d$ contains exactly one node from every set $L_i$ such that the node in $L_d$ is an unmatched node in $V_H$.

We use a variant of the greedy set cover algorithm to find the set  $C$ covering all the shortest augmenting paths in $H$. In order to apply the greedy set cover algorithm, we need to know the sizes of the sets $P_v$, i.e., for every node $v$, we need to know in how many augmenting paths of length $d$ the node $v$ is contained. To compute this number, we apply an algorithm that was first developed in \cite{lotker15} and later refined in \cite{yehuda17}. The following lemma summarizes the result of \cite{yehuda17,lotker15}.

\begin{lemma}\label{lemma:pathcounting}\cite{yehuda17,lotker15}
	Let $H=(U_H\cup V_H, E_H)$ be a bipartite graph of maximum degree at most $\Delta$ 
	and $M_H$ be a matching of $H$. There is a deterministic $O(d^2)$-round \CONGEST\ algorithm to compute the number of shortest augmenting paths of length $d$ passing through every node $v\in U_H\cup V_H$. 
\end{lemma}
\begin{proof}
	Recall  that we assume that the bipartition into $U_H$ and $V_H$ is given. The algorithm to compute the numbers consists of two phases. In a first phase, for every node $v\in U_H\cap V_H$, the algorithm computes the number $x_v$ of shortest alternating paths starting at a node in $L_0$ and ending in $v$. This can be done by a simple top-down aggregation algorithm by going over the different layers $L_0$, $L_1$, \dots, $L_d$, one by one. For a node $v\in L_i$, the number of such paths is exactly the sum of these numbers for $v$'s neighbors in layer $L_{i-1}$. For the unmatched nodes $v\in L_d$, this already gives the desired number of augmenting paths of length $d$ containing $v$. For the other nodes, the numbers can now be computed in adding a bottom-up phase, where we go through the layers $L_i$ in reverse order. For some $i<d$, consider some node $u\in L_i$ and the neighbors $v_1,\dots,v_h$ of $u$ in the next layer $L_{i+1}$. For each of the node $v$ let $p_{v}$ be the number of shortest augmenting paths containing $v_i$. Then, in the bottom-up process, the value of $p_u$ can be computed as $p_u = \sum_{i=1}^hp_{v_i}\cdot x_u/x_{v_i}$. In Claim B.5 of \cite{yehuda17}, it is shown that this algorithm computes the correct number $p_v$ for each node $v\in U_H\cup V_H$.
	
	For computing the numbers $p_v$, we only need to do two passes through the levels $L_0,\dots,L_d$. If the nodes could send arbitrarily large messages, this would require $2d$ rounds. Since the graph can have maximum degree $\Delta$, the number of alternating paths of length at most $d$ passing through a node $v$ can be at most $\Delta^d$. In the algorithm, we therefore have to communicate integers between $1$ and $\Delta^d$, and thus numbers that can be represented with $O(d\log \Delta)$ bits. Communicating a single such number might require up to $O(d)$ rounds in the \CONGEST model.
\end{proof}

We can now use this path counting method to find a small set $S$ of nodes that covers all augmenting paths of length $d$. We start with an empty set $C$. The algorithm then works in $O(d\log\Delta)$ phases $i=1,2,3,\dots$, where in phase $i$, we add unmatched nodes $v$ and matching edges $e$ to $C$ such that are still contained in at least $\Delta^d/2^i$ remaining paths. In order to obtain a polylogarithmic running time, we need to add nodes and edges to $C$ in parallel.  In order to make sure that we do not cover the same path twice, when adding nodes and edges in parallel, we essentially iterate through the $d$ levels in each phase. The details of the algorithm are given in the following.

\begin{center}
	\begin{minipage}{1.0\linewidth}
		\vspace{-8pt}
		\begin{mdframed}[hidealllines=false, backgroundcolor=gray!10]
			\textbf{Covering Paths of Length \boldmath$d$: Phase $i\geq 1$}\\[1mm]
			Iterate over all odd levels $\ell=1,3,\dots,d$:
			\vspace*{-2mm}
			\begin{enumerate}
				\item Count the number of augmenting paths of length $d$ passing through each of the remaining nodes and edges.
				\item If $\ell\in\set{1,d}$, for all remaining nodes $v\in L_\ell$ that are in $p_v\geq \Delta^d/2^i$ different augmenting paths of length $d$, add $v$ to $C$ and remove $v$ and its incident edges from $G_H$ for the remainder of the algorithm.
				\item If $\ell\in\set{2,\dots,d-1}$, for all remaining matching edges $e\in M_H$ connecting two nodes $u\in L_{\ell-1}$ and $v\in L_{\ell}$ that are in $p_e\geq \Delta^d/2^i$ different augmenting paths of length $d$, add $e$ to $C$ and remove $e$ and its incident edges from $G_H$ for the remainder of the algorithm.
			\end{enumerate}
			Define $S_H$ to contain every node in $C$ and both nodes of every edge in $C$.
		\end{mdframed}
	\end{minipage}
	\vspace{-8pt}
\end{center}

\smallskip

\begin{lemma}\label{lemma:onephase}
	Let $\delta\in (0,1)$ and assume that $M_H$ is a $(1-\delta)$-approximate matching of the bipartite graph $H$ of maximum degree at most $\Delta$ . Then, the set $S_H$ selected by the above algorithm has size at most $ \alpha_d \delta  \cdot\OPT_H$, where $\alpha_d=2(d+3)(1+d \ln\Delta)$ and $\OPT_H$ is the size of a maximum matching and thus of a minimum vertex cover of $H$. The time complexity of the algorithm in the \CONGEST model is $O(d^4\log \Delta)$.
\end{lemma}
\begin{proof}
	We first look at the time complexity of the algorithm in the
	\CONGEST model. The algorithm consists of $O(d\log\Delta)$ phases,
	in each phase, we iterate over $O(d)$ levels and in each of these
	iterations, the most expensive step is to count the number of
	augmenting paths passing through each node and edge. By Lemma
	\ref{lemma:pathcounting}, this can be done in time $O(d^2)$,
	resulting in an overall time complexity of $O(d^4\log \Delta)$.
	
	For each free node $v\in U_H\cup V_H$ and for each matching edge
	$e\in M_H$, let $p_v$ and $p_e$ be the number of (uncovered) augmenting paths of
	length $d$ passing through $v$ and $e$, respectively. We will next
	show that our algorithm is simulating a version of the standard
	sequential greedy set cover algorithm. When applying the sequential
	greedy algorithm, in each step, we would need to choose a set $P_v$
	or $P_e$ of paths that maximizes the number of uncovered augmenting
	paths of length $d$ the set covers. We will see that we
	essentially relax the greedy step and we obtain an algorithm that is
	equivalent to a sequential algorithm that always picks a set of paths that
	contains at least half as many uncovered paths as possible. To show
	this, we first show that for each phase $i$, at the beginning of the
	phase, we have $p_v,p_e\leq \Delta^d/2^{i-1}$ for all unmatched nodes $v$ and
	matching edges $e$. For the sake of contradiction, assume that this is not
	the case and let $i'$ be the first phase, in which it is not true.
	Because every node and edge can be contained in at most $\Delta^d$
	augmenting paths of length $d$, the statement is definitely true
	for the first phase and we therefore have $i'>1$. We now consider
	phase $i'-1$. In each phase, by iterating over all odd levels
	$\ell=1,3,\dots,d$, we iterate over all unmatched nodes $v\in
	U_H\cup V_H$ and all matching edges $e\in M_H$ that are contained in
	some augmenting path of length $d$. For each of them, we add the
	corresponding set $P_v$ or $P_e$ to the set cover if we still have
	$p_v\geq \Delta/{2^{i'-1}}$ or $p_e\geq \Delta/2^{i'-1}$. At the end
	of phase $i'-1$, we therefore definitely have $p_v,p_e<
	\Delta/2^{i'-1}$ for all nodes $v$ and matching edges $e$, which
	contradicts the assumption that at the beginning of phase $i'$, it
	is not true that $p_v,p_e\leq
	\Delta/2^{i'-1}$ for all such $v$ and $e$. Because in each phase
	$i$, we only add set $P_v$ and $P_e$ that are contained in at least
	$\Delta/2^{i}$ uncovered paths, we clearly always pick sets that
	cover at least half as many uncovered paths as the best current set.
	Note also that because we iterate through the levels and only add
	sets for nodes or edges on the same level in parallel, the set that
	we add in parallel cover disjoint sets of paths. The algorithm is
	therefore equivalent to a sequential algorithm that adds the sets in
	each parallel step in an arbitrary order.
	
	Now, we will show that we remove at most
	$2(d+3)(1+d\ln\Delta) \delta \cdot \OPT_H$ nodes from graph $H$.
	Indeed, approximating the set cover problem using the standard
	greedy algorithm gives a $(1+\ln(s))$ approximation to the solution,
	where $s$ is the cardinality of the largest set. If we
	relax the greedy step by at least a factor of two, as our algorithm
	does, a standard analysis implies that we still get a $2(1+\ln
	s)$-approximation of the corresponding 
	minimum set cover problem, where $s$ is still defined as the cardinality of
	the largest set. In our case, the largest set $P_v$ or $P_e$ is $s\leq\Delta^d$. Now if
	the solution to the set cover problem using this greedy version
	algorithm is $S_H$ and the optimal solution of the set cover problem
	is $S^*$, then $|S^*| \leq |S_H| \leq 2 (1+d\ln\Delta)|S^*|$. Recall
	that $P_e$ corresponds to a matched edge and by step 3 in our
	algorithm, both of these matched nodes are removed from the graph $H$. Hence, we
	remove up to 
	$2|S_H| \leq 4 (1+d \ln\Delta)|S^*|$ nodes from $H$. 
	
	Next, we give an upper bound to $|S^*|$, which will finish up our
	proof. Recall that a solution to our set cover problem is a
	set of matched edges $S_e$ and a set of unmatched nodes $S_v$ that
	cover all augmenting paths of length $d$ in $H$, i.e., all paths in $\mathcal{P}$.
	Luckily, there is a simple solution to the given set cover problem
	that allows us to upper bound $|S^*|$. We just select a maximal set $P$ of vertex-disjoint
	augmenting paths of length $d$ and we consider all the unmatched nodes
	and matched edges on these paths to be our solution $S'$, where
	$|S'|= \frac{d+3}{2}|P|$. Clearly, $S'$ is a set cover (and thus
	$|S^*|\leq |S'|$), as otherwise there would be an augmenting path of length $d$ that is not covered by $S'$. This path has to be vertex-disjoint from all the paths in $P$, which is a contradiction to the assumption that $P$ is a maximal set of vertex-disjoint augmenting paths of length $d$.  Let
	$|M_H^*|$ denote the maximum cardinality of a matching of graph $H$.
	Now, since $M_H$ is a $(1-\delta)$-approximate matching, we can
	clearly have at
	most $\delta |M_H^*|$ vertex-disjoint augmenting paths of at most length $d$. Hence, the size of $P$ can never exceed
	$\delta |M_H^*|$ i.e. $|P| \leq \delta |M_H^*| $. Thus,
	$|S^*|\leq |S'| \leq \frac{d+3}{2} \delta |M_H^*| $. Hence, we
	remove at most
	$2 |S_H| \leq 4 (1+d \ln\Delta)|S'| \leq 4 (1+d \ln\Delta)
	\frac{d+3}{2} \delta |M_H^*| \leq 2(d+3)(1+d\ln\Delta)
	\delta |M_H^*|= 2(d+3)(1+d\ln\Delta) \delta \cdot \OPT_H $
	nodes from graph $H$.
\end{proof}

By iterating over the lengths of shortest paths, we now directly get the following lemma.

\begin{lemma}\label{lemma:S1bound}
	Let $G=(U\cup V, E)$ be a bipartite graph, let $k\geq 1$ be an integer parameter, and assume that $M$ is a $(1-\delta)$-approximate matching of $G$ for some $\delta\in[0,1]$. Further, let $\OPT$ be the size of a minimum vertex cover of $G$. If the bipartition of the nodes of $G$ into $U$ and $V$ is given, there is an $O(k^5\log \Delta)$-time algorithm to compute a node set $S_1 \subseteq U\cup V$ of size at most $4k(k+1)(1+2k\ln \Delta)\delta\cdot\OPT$ such that in the induced subgraph $G[U\cup V\setminus S]$, there is no augmenting path of length at most $2k-1$ w.r.t.\ the matching $\bar{M}$, where $\bar{M}\subseteq M$ consists of the edges of $M$ that connect two nodes in $U\cup V\setminus S_1$.
\end{lemma}
\begin{proof}
	To compute $S_1$ we start with an empty set $S_1$, and we add nodes to $S_1$ in $k$ stages $i=1,\dots,k$. Let $F_i$ be the set of nodes that are added to $S_1$ in stage $i$. We at all time use $\bar{M}\subseteq M$ to denote the set of edges in $M$ that connect two nodes in $U\cup V\setminus S_1$. We will show that at the beginning of each stage $i$, the shortest augmenting path length in $G[U\cup V\setminus S_1]$ w.r.t.\ matching $\bar{M}$ is at most $2i-1$. The node set $F_i$ of stage $i$ is then selected by using the above algorithm applied to the graph $H=G[U\cup V \setminus S_1]$ with the current matching $\bar{M}$ and $d=2i-1$. To be able to apply the above algorithm in stage $i$, we need to make sure that at the beginning of each stage $i$, the shortest augmenting path length is at least $2i-1$. For contradiction, assume that this is not the case and assume that $i$ is the first stage, where there is an augmenting path of length less than $2i-1$ at the beginning of stage $i$. For $i=1$, the shortest augmenting path length is clearly at least $2i-1=1$ and we can therefore assume that $i>1$. Suppose at the beginning of stage $i$, there exists an augmenting path  in $G[U\cup V\setminus S]$ w.r.t.\ matching $\bar{M}$ of length $2j-1$ for some $j<1$. Let us focus on the two unmatched nodes at the ends of this path. If both these nodes were already unmatched at the beginning of stage $j$, the algorithm of stage $j$ would add at least one node of the path to $F_j$ and thus to $S_1$. This therefore is a contradiction to the assumption that all the nodes of the path are still present at the beginning of stage $i$. We therefore know that at least one of the two unmatched nodes of the path must have been matched at the beginning of the algorithm. However, this also cannot be because the algorithm guarantees that whenever we add one node of a matching edge to $S_1$, then we also add the other node of this matching edge to $S_1$ at the same time. The algorithm can therefore never create unmatched nodes. We therefore know that at the beginning of each stage $i$, the shortest augmenting path length is at least $2i-1$ and we can therefore apply Lemma  \ref{lemma:onephase} to select a set of nodes $F_1$ to cover all augmenting paths of length $2i-1$ in stage $i$. By Lemma  \ref{lemma:onephase}, there is an $O(i^4\log \Delta)$-round \CONGEST algorithm to compute the set $F_i$ of nodes that are added to $S_1$ in stage $i$, such that $|F_i|\leq \alpha_d \delta \cdot\OPT$. And after adding $F_i$ to $S_1$, the graph $G[U\cup V\setminus S_1]$ has no augmenting paths of length at most $2i-1$ w.r.t.\ $\bar{M}$. For the size of $S_1$ after all $k$ stages, we therefore obtain
	\begin{eqnarray*}
		\label{eq:1}
		|S_1| & \leq & \delta \cdot \OPT \cdot \sum_{i=1}^{k} \alpha_{2i-1}\\
		& = & \delta \cdot \OPT \cdot \sum_{i=1}^{k} 2(2i-1+3)(1+(2i-1) \ln\Delta)\\
		& \leq & 4k(k+1)(1+2k\ln \Delta) \cdot \delta \cdot \OPT.
	\end{eqnarray*}
	By Lemma  \ref{lemma:onephase}, the overall running time of all $k$ stages is $k\cdot O(k^4\log \Delta)=O(k^5\log \Delta)$.
\end{proof}

We now have everything that we need to also get a deterministic $O\big(D+\poly\big(\frac{\log n}{\eps}\big)\big)$-time \CONGEST algorithm for computing a $(1+\eps)$-approximate solution for the MVC problem in bipartite graphs.

\begin{theorem}\label{approxMVC:det}
	Let $G=(V,E)$ be a bipartite graph of diameter $D$ and maximum
	degree $\Delta$ and let $G'=(V',E')$ be a subgraph of $G$. For
	$\epsilon\in (0,1]$, there is a deterministic algorithm that
	gives a $(1+\epsilon)$-approximate minimum vertex cover of
	graph $G'$ in $O(D+ \frac{log^4 n}{\epsilon^8})$ rounds in the \CONGEST model on $G$. 
\end{theorem} 
\begin{proof}
	As a first step, we choose a sufficiently small parameter $\delta>0$
	and we compute a $(1-\delta)$-approximate solution $M'$ to the maximum
	matching problem on $G'$ by using the deterministic \CONGEST algorithm of \cite{ahmadi18}.
	For computing such a matching, the algorithm of \cite{ahmadi18}
	has a time complexity of $O\big(\frac{\log^2\Delta +\log^*
		n}{\delta} + \frac{\log\Delta}{\delta^2}\big)=O\big(\frac{\log^2
		n}{\delta^2}\big)$. Let $k':=\lceil 2/\eps\rceil$ as discussed above. By Lemma \ref{lemma:S1bound}, there is a value
	$\alpha=4k'(k'+1)(1+2k'\ln\Delta)= O(k'^3\log\Delta)$ such that we can find a set
	$S_1\subseteq V'$ of size $|S_1|=\alpha\delta\OPT$, where $\OPT$ is
	the size of a minimum vertex cover of $G'$, such that the
	following is true. The set $S_1$ can be computed in time
	$O(k'^5\log\Delta)=O\big(\frac{\log n}{\eps^5}\big)$.
	Let $G''=G'[V'\setminus S_1]$ be the induced
	subgraph of $G'$ after removing all the nodes in $S_1$ and let
	$M''$ be the subset of the edges in $M'$ that connect two nodes
	in $V'\setminus S_1$ (i.e., $M''$ is a matching of $G''$). Then, the
	graph $G''$ has no augmenting paths of length at most $2k'-1$. By
	using Lemma \ref{lemma:diameter1}, we can therefore compute a
	$(1+1/k')$-approximate vertex cover $S_2$ (and thus a $(1+\eps/2)$-approximate
	vertex cover) of $G''$ in time $O(D+k')=O(D+1/\eps)$. Because a
	minimum vertex cover of $G''$ is clearly not larger than a minimum
	vertex cover of $G'$, we therefore have $|S_2|\leq
	(1+\eps/2)\cdot\OPT$. Note that
	$S_1\cup S_2$ is a vertex cover of $G'$. The size of $S_1\cup S_2$
	can be bounded as $|S_1\cup S_2| \leq \delta\alpha\cdot\OPT +
	(1+\eps/2)\cdot\OPT$. In order to make sure that this is at most
	$(1+\eps)\cdot\OPT$, we have to choose $\delta\leq \eps/(2\alpha)$.
	The time complexity to compute the initial matching $M'$ of $G'$ is
	therefore
	$O\big(\frac{\log^2 n}{\delta^2}\big)=O\big(\frac{\log^4
		n}{\eps^8}\big)$. 
\end{proof}

\section{Polylogarithmic-Time Algorithms}
\label{sec:polylog}

We next show how we can use the algorithms of the previous section
together with existing low-diameter graph clustering techniques to
obtain polylogarithmic-time approximation schemes for the minimum
vertex cover algorithm in the \CONGEST model. First we describe a
general framework for achieving a $(1+\epsilon)$-approximate minimum
vertex cover $C$ of unweighted bipartite graphs via an efficient
algorithm in the \CONGEST model based on a given clustering with some
specific properties (cf. Section ~\ref{sec:model} for the corresponding
definitions). We will do so by proving the following lemma. Note that
our general framework applies to both the randomized and the
deterministic case. 

\begin{lemma} 
	\label{approxMVC:givenclustering} Let $G=(V,E)$ be a
	bipartite graph and assume that we are given a maximal matching $M$
	of $G$. We define edge weights $w(e)\in \{0,1\}$ such that $w(e)=1$
	if and only if $e \in M$. Further, assume that w.r.t.\ those edge
	weights, we are given a $(1-\eta)$ dense, $3$-hop separated, and
	$(c,d)$-routable clustering of $G$, for some $\eta\in(0,1]$ and some
	positive integers $c,d>0$. Then, for any $\psi \in(0,1]$, we can
	find a $(1+2\eta+\psi)$-approximate minimum vertex cover by a
	deterministic \CONGEST algorithm in
	$O\big(c \cdot \big(d+\poly \frac{\log n}{\psi}\big)\big)$ rounds
	and by a randomized \CONGEST algorithm in
	$O\big(c \cdot \big(d+\frac{\log n}{\psi^3}\big)\big)$ rounds,
	w.h.p.  
\end{lemma}

\begin{proof}
	Let $\{S_1,S_2,...,S_t\}$ be the collection of clusters
	of the given $3$-hop separated, $(1-\eta)$-dense clustering. Define
	$E'$ to be the set of edges for which both endpoints are located
	outside clusters and let $E''$ to be the set of edges where exactly
	one of the endpoints is outside clusters. We also say that $e$
	is an edge outside clusters if it is in $E' \cup E''$. Further,
	let $X$ to be the set of all matched nodes (w.r.t.
	the given maximal matching $M$) that are outside
	clusters. Note that since $M$ is a maximal matching, any
	edge in $E'$ is necessarily incident to at least one matched node of
	$M$.  Therefore, when adding the set $X$ to the vertex cover
	$C$, we cover all edges in $E'$ and
	possibly some extra edges in $E''$. Now since $G$ is
	$(1-\eta)$-dense, then at most $\eta |M|$ matched edges are outside
	clusters, and when assuming that $|M^*|$ is the size of a maximum
	matching of $G$, we can deduce that
	$|X| \leq 2\eta|M| \leq 2\eta |M^*|= 2\eta \OPT$, where $\OPT$ is
	the size of a minimum vertex cover of $G$. Next, we extend
	each cluster $S_i$ by at most one hop in radius as follows. For
	every edge $\set{u,v}\in E''$ such that $u\in S_i$ and $v\not\in
	S_i$, we add the edge $\set{u,v}$ and node $v$ to the cluster.
	Let $\{S'_1,S'_2,...,S'_t\}$ be the new collection of extended
	clusters. All edges of $G$ that are not already covered by $X$ are
	now inside some cluster. In addition, we grow the height of each
	cluster tree $T_i$ by at most one hop so that they include the new
	cluster nodes. We denote the new extended trees by
	$T'_i$. Note that clearly, each edge in $E$ is still in at most $c$
	trees. Hence, the new collection of extended clusters are now $2$-hop
	separated 
	and $(c,d+1)$-routable.
	
	For each cluster $S_i'$, let $G_i'$ be the graph consisting of the
	nodes and edges of the cluster. We note that because the clusters
	are $1$-hop separated, the graphs $G_i'$ are vertex and edge
	disjoint. In addition, for each cluster $S_i'$, we define the graph
	$G_i$ as the union of $G_i'$ and the tree $T_i'$. Because the
	clustering is $(c,d+1)$-routable, it follows that every edge of $G$
	is used by at most $c$ of the graph $G_i$ and that the diameter of
	each graph $G_i$ is at most $d+1$. To obtain a vertex cover of all
	edges of $G$, we now compute a $(1+\psi)$-approximate minimum vertex
	cover $C_i$ for each extended cluster graph $G_i'$ by running the
	algorithms described in Theorems \ref{approxMVC:rand} and \ref{approxMVC:det}. We do
	this for all clusters in parallel. For each cluster $S_i'$, we use
	$G_i$ and $G_i'$ as the graphs $G$ and $G'$ in
Theorems \ref{approxMVC:rand} and \ref{approxMVC:det}. Because each edge is contained
	in at most $c$ graphs $G_i$, we can in parallel run $T$-round
	algorithms in all graphs $G_i$ in time $c\cdot T$. The time
	complexities therefore follow directly as claimed from the
	respective time complexities in Theorems \ref{approxMVC:rand} and \ref{approxMVC:det}.
	
	We define $Y:= \bigcup_{i=1}^{t}C_i$. Because every edge of $G$ that
	is not covered by the nodes in $X$ is inside one of the clusters
	$S_i'$, clearly, the set $X\cup Y$ is a vertex cover of $G$. We
	already showed that $|X|\leq 2\eta\OPT$. To bound the size of
	$X\cup Y$, it remains to bound the size of $Y$. Let $\OPT_i$ be the
	size of an optimal vertex cover of $G_i'$. Because the cluster
	graphs $G_i'$ are vertex-disjoint, all edges in $G_i'$ clearly have
	to be covered by some node of the cluster $S_i'$ and thus edges in
	different clusters have to be covered by disjoint sets of nodes. If
	$\OPT$ is the size of an optimal vertex cover of $G$, we thus
	clearly have $\bigcup_{i=1}^t \OPT_i \leq \OPT$. Because $C_i$ is a
	$(1+\psi)$-approximate vertex cover of $G_i'$, we also have
	$|C_i|\leq (1+\psi)\cdot\OPT_i$. Together, we therefore directly get
	that $|Y|\leq (1+\psi)\cdot\OPT$ and therefore
	$|X\cup Y|\leq (1+2\eta+\psi)\cdot\OPT$.
\end{proof}

In order to prove our two main results, Theorems \ref{thm:randomized} and \ref{thm:deterministic}, we will next show how to
efficiently compute the clusterings that are required for Lemma
\ref{approxMVC:givenclustering}. Both clusterings can be obtained by
minor adaptations of existing clustering techniques.

\subsection{The Randomized Clustering }

We start with describing the randomized clustering algorithm. By using
the exponentially shifted shortest paths approach of Miller, Peng, and
Xu~\cite{MPX13}, we obtain the following lemma.

\begin{lemma} \label{clustering:rand} Let $G=(V,E,w)$ be a weighted
	bipartite graph with non-negative edge weights $w(e)$. For
	$\lambda \in (0,1]$, there is a randomized algorithm that computes a
	$3$-hop separated clustering of $G$ such that w.h.p., the clustering
	is $(1,O(\frac{\log n)}{\lambda})$-routable and can be computed in
	$O(\frac{\log n}{\lambda})$ rounds in the \CONGEST model and such
	that the clustering is $(1-\lambda)$-dense in expectation.
\end{lemma}
\begin{proof} 
	Let $G=(V,E,w)$ be a weighted bipartite graph with non-negative edge
	weights $w(e)$ and assume that $W:=\sum_{w\in E} w(e)$ is the total
	weight of all edges in $G$. Let $\lambda \in (0,1]$. We first run a
	partitioning algorithm using the exponentially shifted shortest
	paths' method of \cite{MPX13}. Each vertex $u$ in $G$ picks shifts
	$\delta_u$ from independent exponential distributions with parameter
	$\sigma=\lambda/4$. For two nodes $u$ and $v$, let the shifted
	distance from $u$ to $v$ be
	$\dist_{-\delta}(u,v):=\dist_G(u,v)-\delta_u$. Each node $v$ is
	assigned to a cluster $S_u$ if the shifted distance
	$\dist_{-\delta}(u,v)$ is minimized among all nodes $u\in V$. This
	algorithm outputs a partition of $V$ into connected clusters (if $v$
	is in cluster $S_u$, then all nodes on a shortest path from $u$ to
	$v$ are also in cluster $S_u$).
	
	After partitioning the nodes into clusters, we shrink all the
	clusters as follows. For every edge $\set{x,y}$ that is between
	clusters, we remove both $x$ and $y$ from their respective
	clusters. For every node $v$ that remains in a cluster, before
	shrinking the clusters all neighbors of $v$ were in the same cluster
	as $v$. Therefore, two nodes $u$ and $v$ in different clusters
	cannot have a common neighbor and therefore we now clearly have a
	3-hop separated clustering. We next bound the number of matching
	edges outside clusters (i.e., edges for which not both endpoints are
	inside a cluster).  For each such edge $\set{u,v}$, we know that at
	least one of the two nodes $u$ or $v$ has a neighbor $w$ that was
	initially assigned to a different cluster. W.l.o.g., assume that $v$
	and $w$ are neighbors and that those two nodes were assigned to
	different clusters. Assume that $v$ is initially assigned to cluster
	$S_x$ and $w$ is initially assigned to cluster $S_y$ (where
	$x\neq y$). We then know that
	$\dist_{-\delta}(x,v)\leq \dist_{-\delta}(y,w)+1$ and
	$\dist_{-\delta}(y,w)\leq \dist_{-\delta}(x,v)+1$, and therefore
	$|\dist_{-\delta}(x,v) - \dist_{-\delta}(y,w)|\leq 1$. This also
	implies that $|\dist_{-\delta}(x,v) - \dist_{-\delta}(y,v)|\leq 2$,
	i.e., the difference between the smallest and the second smallest
	shifted distance for $v$ is at most $2$. In [Lemma 4.3, \cite{
		MPX13}], it is shown that for every node $v$, the probability that
	the two smallest shifted distances for $v$ differ by at most $2$ is
	bounded by $2\sigma$. By a union bound over the two nodes $u$ and
	$v$ of the edge $\set{u,v}$, we therefore get that the probability
	that the edge is outside a cluster (after shrinking clusters) is at
	most $4\sigma$. By linearity of expectation, we therefore
	immediately get that the expected total weight of all the edges
	outside clusters is at most $4\sigma W=\lambda W$. Thus, our
	clustering is $(1-\lambda)$-dense in expectation. 
	
	Furthermore in \cite{MPX13}, it is shown that with high probability,
	the strong diameter of each cluster is bounded by
	$O(\frac{\log n}{\sigma})$. By just computing a BFS tree of each
	cluster, we therefore directly obtain that the computed clustering
	is $(1,O(\frac{\log n}{\lambda}))$-routable clustering,
	w.h.p. Finally notice that their partition algorithm of \cite{MPX13}
	can be directly implemented in $O(\frac{\log n}{\lambda})$ \CONGEST
	rounds. Each node can sort the values of the shifted distances it
	receives in each round and always just forward the smallest one
	among them to the neighbors. Also note that the value of an exponential random variable is an arbitrary real number and therefore cannot be represented by $O(\log n)$ bits. It is however clearly sufficient to round each of the exponential random variables such that the relative accuracy is
	$(1\pm 1/n^c)$. If we choose the constant $c$ sufficiently large, w.h.p.,
	the random variable remain distinct and the relative order of the shifted distances does not change.
\end{proof}

We now have everything that we need to prove our first main result,
our randomized polylogarithmic-time approximation scheme for the MVC
problem in bipartite graphs.

\begin{proof}[\textbf{Proof of Theorem \ref{thm:randomized}}]
	Let $G=(V,E)$ be the given bipartite graph for which we want to
	approximate the MVC problem. We first compute a maximal matching $M$
	of $G$, which we can for example do by using Luby's algorithm~\cite{alon86,luby86} in
	$O(\log n)$ rounds. By using $M$, we then apply Lemma \ref{clustering:rand} with
	$\lambda=\eps/4$ to obtain a $3$-hop separated $\big(1,
	O\big(\frac{\log n}{\eps}\big)\big)$-routable clustering that is
	$(1-\eps/4)$-dense in expectation. The time for computing the
	clustering is $O\big(\frac{\log n}{\eps}\big)$, w.h.p. By applying
Lemma \ref{approxMVC:givenclustering} with $\eta=\eps/4$ and
	$\psi=\eps/2$, we then get a vertex cover of $G$ in 
	$O\big(\frac{\log n}{\eps^3}\big)$ \CONGEST rounds such that the expected size of
	the vertex cover is at most $(1+\eps)\cdot\OPT$, where $\OPT$ is the
	size of a minimum vertex cover of $G$. This concludes the proof of
	the theorem.
\end{proof}

\subsection{The Deterministic Clustering }

We obtain the deterministic version of the necessary clustering by
adapting the construction of a single color class of the recent
efficient deterministic network
decomposition algorithm of Rozho\v{n} and
Ghaffari~\cite{polylogdecomp}.

\begin{lemma} \label{clustering:det}
	Let $G=(V,E,w)$ be a weighted bipartite graph with non-negative edge
	weights $w(e)\in \set{0,1}$.  For $\lambda \in (0,1]$, there is a
	deterministic algorithm that computes an $(1-\lambda)$-dense,
	$3$-hop separated, and $\big(O(\log n),O\big(\frac{\log^3n}{\lambda}\big)\big)$-routable clustering
	of $G$ in $\poly\big(\frac{\log n}{\lambda}\big)$ rounds in the \CONGEST model.
\end{lemma}
\begin{proof}
	We assume that $W:=\sum_{w\in E} w(e)$ is the total weight of all
	edges in $G$. Let $\lambda \in (0,1]$. We adapt the weak diameter
	network decomposition algorithm of Rozho\v{n} and
	Ghaffari~\cite{polylogdecomp} applied to the graph $G^2$ in the
	\CONGEST model. When applied to $G^2$, Theorem 2.12 of
	\cite{polylogdecomp} shows that the algorithm of
	\cite{polylogdecomp} computes a decomposition of the nodes $V$ into
	clusters of $O(\log n)$ colors such that any two nodes in different
	clusters of the same color are at distance at least $3$ from each
	other (in $G$). Each cluster is spanned by a Steiner tree of
	diameter $O(\log^3 n)$ such that each edge of $G$ is used by at most
	$O(\log n)$ different Steiner trees for each of the $O(\log n)$
	color classes. For our purpose, we only need to construct the first
	color class of this decomposition. For the first color class, the
	proof of Theorem 2.12 of \cite{polylogdecomp} implies that the
	clusters of the first color are $3$-hop separated and that they
	contain a constant fraction of all the nodes. We need to adapt the
	construction of the first color class of the algorithm of
	\cite{polylogdecomp} in two ways. In the following, we only sketch
	these changes.
	
	First, we adapt the algorithm so that it can handle weights. In the
	following, we define node weight $\nu(v)\geq 0$ as follows. For each
	node $v$, we define $\nu(v)$ as the sum of the weights $w(e)$ of the
	edges $e$ that are incident to $v$. Note that this implies that the
	total weight of all the nodes is $2W$ and that the total weight of
	all the nodes that are not clustered is an upper bound on the total
	weight of all the edges outside clusters (i.e., all the edges, where
	at most one endpoint is inside a cluster). In the algorithm of
	\cite{polylogdecomp}, the clustering is computed in different
	steps. In each step, some nodes request to join a different cluster
	and a cluster accepts these requests if the total number of nodes
	requesting to join the cluster is large enough compared to the total
	number of nodes already inside the cluster. If a cluster does not
	accept the requests, the requesting nodes are deactivated and will
	not be clustered. The threshold on the number of requests required to
	accept the requests is chosen such that in the end the weak diameter of
	the clusters is not too large and at the same time, only a constant
	fraction of all nodes are deactivated and thus not clustered. In our
	case, we do not care how many nodes are clustered and unclustered,
	but we care about the total weight of nodes that are clustered and
	unclustered. The analysis of \cite{polylogdecomp} however directly
	also works if we instead compare the total weight of the nodes that request
	to join a cluster with the total weight of the nodes that are
	already inside the cluster. If the node weights are polynomially
	bounded non-negative integers (which they are in our case), the
	asymptotic guarantees of the construction are exactly the same. In
	this way, we can make sure to construct $(O(\log n), O(\log^3
	n))$-routable, $3$-hop separated clusters such that a constant
	fraction of the total weight of all the nodes is inside clusters.
	
	As a second change, in order to make sure that the clustering is
	also $(1-\lambda)$-dense, we need to guarantee that the total weight of
	the nodes that are unclustered is at most a $\lambda/2$-fraction of the
	total weight of all the nodes. We can guarantee this, by adapting
	the threshold for accepting nodes to a cluster. We essentially have
	to multiply the threshold by a factor $\Theta(\lambda)$ to make sure
	that this is the case. This increases the maximal possible cluster
	diameter by a factor $O(1/\lambda)$ and it increases the total
	running time by a factor $\poly(1/\lambda)$.
\end{proof}

\noindent\textit{Remark:} In the above lemma, we assumed for simplicity that the
edge weights are either $0$ or $1$. The construction however directly
also works in the same way and with the same asymptotic guarantees if
the edge weights are polynomially bounded non-negative integers. With
some simple preprocessing, one can also obtain the same asymptotic
result for arbitrary non-negative edge weights.

\smallskip

In a similar way as we proved Theorem \ref{thm:randomized}, we can now also
prove our second main result, our deterministic polylogarithmic-time
approximation scheme for the MVC problem in bipartite graphs.

\begin{proof}[\textbf{Proof of Theorem \ref{thm:deterministic}}]
	Let $G=(V,E)$ be the given bipartite graph for which we want to
	approximate the MVC problem. We first compute a maximal matching $M$
	of $G$, which we can do by using the algorithm of Fischer~\cite{rounding} in
	$O(\log^2\Delta \cdot \log n)$ deterministic rounds in the \CONGEST model. By using $M$, we then apply Lemma \ref{clustering:det} with
	$\lambda=\eps/4$ to obtain a $(1-\eps/4)$-dense, $3$-hop separated $\big(O(\log n),
	\poly\big(\frac{\log n}{\eps}\big)\big)$-routable clustering. By
	Lemma \ref{clustering:det}, the time for computing the
	clustering in the \CONGEST model is $\poly\big(\frac{\log n}{\eps}\big)$. By applying
Lemma \ref{approxMVC:givenclustering} with $\eta=\eps/4$ and
	$\psi=\eps/2$, we then get a $(1+\eps)$-approximate vertex cover of $G$ in 
	$\poly\big(\frac{\log n}{\eps}\big)$ \CONGEST rounds, which
	completes the proof of the theorem.
\end{proof}


\bibliographystyle{alpha}
\bibliography{references}

\newcommand{\etalchar}[1]{$^{#1}$}
\begin{thebibliography}{BCM{\etalchar{+}}20}

\bibitem[ABI86]{alon86}
N.~Alon, L.~Babai, and A.~Itai.
\newblock A fast and simple randomized parallel algorithm for the maximal
  independent set problem.
\newblock {\em Journal of Algorithms}, 7(4):567--583, 1986.

\bibitem[{\AA}FP{\etalchar{+}}09]{AstrandFPRSU09}
M.~{\AA}strand, P.~Flor{\'{e}}en, V.~Polishchuk, J.~Rybicki, J.~Suomela, and
  J.~Uitto.
\newblock A local 2-approximation algorithm for the vertex cover problem.
\newblock In {\em Proc.\ 23rd Symp.\ on Distributed Computing (DISC)}, pages
  191--205, 2009.

\bibitem[AKO18]{ahmadi18}
M.~Ahmadi, F.~Kuhn, and R.~Oshman.
\newblock Distributed approximate maximum matching in the {CONGEST} model.
\newblock In {\em Proc.\ 32nd Symp.\ on Distributed Computing (DISC)}, pages
  6:1--6:17, 2018.

\bibitem[BCD{\etalchar{+}}19]{bachrach_podc19}
N.~Bachrach, K.~Censor{-}Hillel, M.~Dory, Y.~Efron, D.~Leitersdorf, and A.~Paz.
\newblock Hardness of distributed optimization.
\newblock In {\em Proc.\ 38th {ACM} Symp.\ on Principles of Distributed
  Computing (PODC)}, pages 238--247, 2019.

\bibitem[BCGS17]{yehuda17}
R.~Bar{-}Yehuda, K.~Censor{-}Hillel, M.~Ghaffari, and G.~Schwartzman.
\newblock Distributed approximation of maximum independent set and maximum
  matching.
\newblock {\em CoRR}, abs/1708.00276, 2017.
\newblock Conference version at PODC 2017.

\bibitem[BCM{\etalchar{+}}20]{Bar-YehudaCMPP20}
R.~Bar{-}Yehuda, K.~Censor{-}Hillel, Y.~Maus, S.~Pai, and S.~V. Pemmaraju.
\newblock Distributed approximation on power graphs.
\newblock In {\em Proc.\ 39th {ACM} Symp.\ on Principles of Distributed
  Computing (PODC)}, pages 501--510, 2020.

\bibitem[BCS16]{yehuda16}
R.~Bar{-}Yehuda, K.~Censor{-}Hillel, and G.~Schwartzman.
\newblock A distributed (2+{\(\epsilon\)})-approximation for vertex cover in
  o(log{\(\delta\)}/{\(\epsilon\)} log log {\(\delta\)}) rounds.
\newblock In {\em Proceedings of the {ACM} Symposium on Principles of
  Distributed Computing (PODC)}, pages 3--8, 2016.

\bibitem[BEKS19]{disc19_optcovering}
R.~Ben{-}Basat, G.~Even, K.~Kawarabayashi, and G.~Schwartzman.
\newblock Optimal distributed covering algorithms.
\newblock In {\em Proc.\ 33rd Symp.\ on Distributed Computing (DISC)}, pages
  5:1--5:15, 2019.

\bibitem[BEPS12]{barenboim12}
L.~Barenboim, M.~Elkin, S.~Pettie, and J.~Schneider.
\newblock The locality of distributed symmetry breaking.
\newblock In {\em Proceedings of 53th Symposium on Foundations of Computer
  Science (FOCS)}, 2012.

\bibitem[BGK{\etalchar{+}}14]{blelloch14}
G.~E. Blelloch, A.~Gupta, I.~Koutis, G.~L. Miller, R.~Peng, and K.~Tangwongsan.
\newblock Nearly-linear work parallel {SDD} solvers, low-diameter
  decomposition, and low-stretch subgraphs.
\newblock {\em Theory Comput.~Syst.}, 55(3):521--554, 2014.

\bibitem[CHKP17]{censorhillel_disc17}
K.~Censor-Hillel, S.~Khoury, and A.~Paz.
\newblock Quadratic and near-quadratic lower bounds for the {CONGEST} model.
\newblock In {\em Proc.\ 31st Symp.\ on Distributed Computing (DISC)}, pages
  10:1--10:16, 2017.

\bibitem[Die05]{koenig_diestel}
R.~Diestel.
\newblock {\em Graph Theory}, chapter 2.1.
\newblock Springer, Berlin, 3rd edition, 2005.

\bibitem[Fis17]{rounding}
M.~Fischer.
\newblock Improved deterministic distributed matching via rounding.
\newblock In {\em Proc.\ 31st Symp.\ on Distributed Computing (DISC)}, pages
  17:1--17:15, 2017.

\bibitem[FMS15]{FMS15}
Uriel Feige, Yishay Mansour, and Robert~E. Schapire.
\newblock Learning and inference in the presence of corrupted inputs.
\newblock In {\em Proc.\ 28th Conf.\ on Learning Theory (COLT)}, pages
  637--657, 2015.

\bibitem[GJN20]{GhaffariJN20}
M.~Ghaffari, C.~Jin, and D.~Nilis.
\newblock A massively parallel algorithm for minimum weight vertex cover.
\newblock In {\em Proc.\ 32nd {ACM} Symp.\ on Parallelism in Algorithms and
  Architectures (SPAA)}, pages 259--268, 2020.

\bibitem[GKM17]{ghaffari2017complexity}
M.~Ghaffari, F.~Kuhn, and Y.~Maus.
\newblock On the complexity of local distributed graph problems.
\newblock In {\em Proc.\ 39th ACM Symp.\ on Theory of Computing (STOC)}, pages
  784--797, 2017.

\bibitem[GKP08]{GrandoniKP08}
F.~Grandoni, J.~K{\"{o}}nemann, and A.~Panconesi.
\newblock Distributed weighted vertex cover via maximal matchings.
\newblock {\em {ACM} Trans.\ Algorithms}, 5(1):6:1--6:12, 2008.

\bibitem[GKPS08]{GrandoniKPS08}
F.~Grandoni, J.~K{\"{o}}nemann, A.~Panconesi, and M.~Sozio.
\newblock A primal-dual bicriteria distributed algorithm for capacitated vertex
  cover.
\newblock {\em {SIAM} J. Comput.}, 38(3):825--840, 2008.

\bibitem[GS14]{goeoes14_DISTCOMP}
M.~G\"o\"os and J.~Suomela.
\newblock No sublogarithmic-time approximation scheme for bipartite vertex
  cover.
\newblock {\em Distributed Computing}, 27(6):435--443, 2014.

\bibitem[HK73]{hopkarp}
J.~E. Hopcroft and R.~M. Karp.
\newblock An $n^{5/2}$ algorithm for maximum matchings in bipartite graphs.
\newblock {\em SIAM Journal on Computing}, 1973.

\bibitem[II86]{itai86}
A.~Israeli and A.~Itai.
\newblock A fast and simple randomized parallel algorithm for maximal matching.
\newblock {\em Inf. Process. Lett.}, 22(2):77--80, 1986.

\bibitem[K\H31]{koenig31}
D.~K\H{o}nig.
\newblock Gr\'{a}fok \'{e}s m\'{a}trixok.
\newblock {\em Matematikai \'{e}s Fizikai Lapok}, 38:116--119, 1931.

\bibitem[KMW04]{lowerbound}
F.~Kuhn, T.~Moscibroda, and R.~Wattenhofer.
\newblock What cannot be computed locally!
\newblock In {\em Proceedings of 23rd ACM Symposium on Principles of
  Distributed Computing (PODC)}, pages 300--309, 2004.

\bibitem[KMW06]{nearsighted}
F.~Kuhn, T.~Moscibroda, and R.~Wattenhofer.
\newblock The price of being near-sighted.
\newblock In {\em Proceedings of 17th Symposium on Discrete Algorithms (SODA)},
  pages 980--989, 2006.

\bibitem[KR08]{MVC_UGChard}
Subhash Khot and Oded Regev.
\newblock Vertex cover might be hard to approximate to within 2-epsilon.
\newblock {\em J. Comput. Syst. Sci.}, 74(3):335--349, 2008.

\bibitem[LPP15]{lotker15}
Z.~Lotker, B.~Patt{-}Shamir, and S.~Pettie.
\newblock Improved distributed approximate matching.
\newblock {\em J.\ ACM}, 62(5):38:1--38:17, 2015.

\bibitem[LS93]{linial93}
N.~Linial and M.~Saks.
\newblock Low diameter graph decompositions.
\newblock {\em Combinatorica}, 13(4):441--454, 1993.

\bibitem[Lub86]{luby86}
M.~Luby.
\newblock A simple parallel algorithm for the maximal independent set problem.
\newblock {\em SIAM Journal on Computing}, 15:1036--1053, 1986.

\bibitem[MPX13]{MPX13}
G.~L. Miller, R.~Peng, and S.~C. Xu.
\newblock Parallel graph decompositions using random shifts.
\newblock In {\em Proc.\ 25th {ACM} Symp.\ on Parallelism in Algorithms and
  Architectures (SPAA)}, pages 196--203, 2013.

\bibitem[Pel00]{peleg00}
D.~Peleg.
\newblock {\em Distributed Computing: A Locality-Sensitive Approach}.
\newblock SIAM, 2000.

\bibitem[RG20]{polylogdecomp}
V.~Rozho\v{n} and M.~Ghaffari.
\newblock Polylogarithmic-time deterministic network decomposition and
  distributed derandomization.
\newblock In {\em Proc.\ 52nd {ACM} Symp.\ on Theory of Computing (STOC)},
  pages 350--363, 2020.

\end{thebibliography}

\end{document}